\documentclass[conference]{IEEEtran}
\IEEEoverridecommandlockouts
\usepackage{cite}
\usepackage{amsmath,amssymb,amsfonts, amsthm}
\usepackage{mathdots}
\usepackage[noend]{algorithm2e}
\usepackage{graphicx}
\usepackage{textcomp}
\usepackage{xcolor}
\usepackage{enumerate}
\newcommand{\Fq}{ \mathbb{F}_{q}}
\newcommand{\K}{ \mathbb{K}}
\newcommand{\Fqx}{\mathbb{F}_{q}\lbrack x \rbrack}
\newcommand{\GCD}{\text{GCD}}
\newcommand{\y}{\boldsymbol{y}}
\newcommand{\f}{\boldsymbol{f}}

\newcommand{\rank}{\text{rank}}
\newtheorem{theorem}{Theorem}
\newtheorem{definition}{Definition}
\newtheorem{proposition}{Proposition}
\newtheorem{remark}{Remark}
\newtheorem{lemma}{Lemma}
\def\BibTeX{{\rm B\kern-.05em{\sc i\kern-.025em b}\kern-.08em
    T\kern-.1667em\lower.7ex\hbox{E}\kern-.125emX}}
\begin{document}

\title{Polynomial Linear System Solving with Errors\\ by Simultaneous Polynomial Reconstruction\\ of Interleaved Reed-Solomon Codes}

\author{\IEEEauthorblockN{Eleonora Guerrini, Romain Lebreton, Ilaria Zappatore}
\IEEEauthorblockA{
\textit{LIRMM, Universit\'e de Montpellier, CNRS}\\
Montpellier, France \\
\textbf{\{guerrini, lebreton, zappatore\}@lirmm.fr}}}

\maketitle

 \begin{abstract}
     In this paper we present a new algorithm for Polynomial Linear System Solving (via evaluation/interpolation) with errors. In this scenario, errors can occur in the black box evaluation step. We improve the bound on the number of errors that we can correct, using techniques inspired by the decoding procedure of Interleaved Reed-Solomon Codes.
\end{abstract}

\section{Introduction}
The problem of decoding a Reed-Solomon code (shortly RS), also known as the Polynomial Reconstruction Problem (PR) has been largely studied in Coding Theory \cite{bw, reedSol1, reedSol3}. 
In \cite{SPR}, Bleichenbacher et al. proposed a new scenario of the PR problem, called Simultaneous Polynomial Reconstruction (SPR). This problem was associated to the decoding of Interleaved Reed-Solomon codes. Instead of the separate reconstruction of each interleaved codeword, the main idea was to correct several codewords simultaneously in order to gain an error correction capability which depends also on the amount of messages received  (interleaving parameter). They  proposed an algorithm that, under some hypotheses on the error distribution, allows to correctly decode Interleaved RS codewords, beyond the unique decoding bound, with a certain probability. This probability depends on the number of errors and on the order of the field of the coefficients.
Interleaved Reed Solomon codes (IRS) are widely studied in the last 20 years. In the original work  \cite{SPR} the key equations for recovering the codewords  are a generalization of  Berlekamp-Welch decoding method for RS.  A more general scenario is due to \cite{CollaborativeIRSDec}  (improved by \cite{DecIRSbeyondJointErrCorrCap}) where codewords are issue of differents RS codes ( namely Heterogeneous IRS ) and the decoding method is based on  Berlekamp-Massey algorithm as for the classical BCH codes. Recently, in \cite{IRSwithPower} by applying Power Decoding method for generating independent key equations \cite{power},  the IRS decoding radius is significantly improved.

\bigskip
The purpose of the present  work is to introduce a new algorithm, inspired by SPR problem, to  solve a full rank consistent linear system $A(x)\boldsymbol{y}=\boldsymbol{b}(x)$ where erroneous evaluations occur.

\noindent
In order to solve this system, a classical technique (see for example \cite{McClellan}) consists in evaluating  in a certain number of points, solving the evaluated system and then interpolating these evaluated solutions. The solution is a vector of rational functions $\frac{\f(x)}{g(x)}$, where $\f$ is  a vector of polynomials and $g(x)$ is the least common denominator. 

\noindent
In \cite{Kalt1, Kalt},  authors studied the problem in a scenario where some evaluations can be erroneous. They introduced an algorithm that recovers the solution by fixing a certain number $L_{BK}$ of evaluation points. This method is a generalization of the Berlekamp-Welch algorithm for  PR. Thus, the error correction capability coincides with the unique decoding bound.

\bigskip
In this work, we generalize the SPR problem to the Simultaneous Rational Function Reconstruction (SRFR) in order to solve a polynomial linear system with errors. 
In the special case where the matrix $A$ is the identity matrix, we observe that our problem reduces to the SPR, or equivalently, to the problem of decoding an Interleaved RS code. Still in this special case, we can apply the decoding technique of Interleaved RS codes and correct more errors than \cite{Kalt1} under the probabilistic hypotheses of \cite{Choko}. 
In order to generalize this, we reexamine the scenario of \cite{Kalt1, Kalt} with a probabilistic assumption. In this context, we introduce a new algorithm that can be seen as the generalization of the decoding algorithm of Interleaved RS codes. 
Our algorithm can correctly reconstruct the vector solution of our system with a smaller number of points $L_{GLZ} \leq L_{BK}$. However it can fail for a small fraction of possible errors as in \cite{Choko}. In our case, the fraction is at most $\frac{dg+e}{q}$ where $dg$ is the degree of the common denominator of the rational function vector, $e$ is the number of errors and $q$ the order of the field.

\section{Polynomial linear system solving with errors}\label{plswitherr}

In \cite{Kalt1} and \cite{Kalt}, authors studied the problem of solving a consistent linear system
\begin{equation}\label{linSystem}
A(x)\boldsymbol{y}(x)=\boldsymbol{b}(x)
\end{equation}
where,
\begin{itemize}
    \item $A(x)$  is a full rank $m \times n$ matrix whose entries are polynomials in $\mathbb{K}[x]$, $\mathbb{K}$ is a field and $m \geq n \geq 1$; 
    \item $\boldsymbol{b}(x)$ is an $m$-th vector of polynomials in $\mathbb{K}[x]$.
\end{itemize}

\noindent
The system admits a solution whose coordinates are rational functions and, since the matrix is full rank, there is a unique solution
$$
\y(x)=\frac{\f(x)}{g(x)}=\left(\frac{f_1(x)}{g(x)}, \ldots, \frac{f_n(x)}{g(x)}\right)
$$
where $g$ is the monic least common denominator, and

\begin{equation}\label{gcd}
\GCD(\f, g)=\GCD(\GCD_i(f_i),g)=1.
\end{equation}

In general, this solution can be found by evaluating the system at a certain number, say $L$, of distinct points $\alpha_l \in \K$ with $l \in \{1, \ldots, L\}$, solving the evaluated system and then interpolating the parametric solution from the evaluated solution \cite{McClellan}.
The authors in \cite{Kalt1} proved that it is possible to reconstruct the solution even if some evaluations are \textit{erroneous}.
They focused on a model where there is a black box that, for any evaluation point $\alpha_l$, provides $A_l \in \K^{m \times n}$ and $\boldsymbol{b}_l \in \K^m$ which may not be equal to $A(\alpha_l)$ and $\boldsymbol{b}(\alpha_l)$. More specifically, the resulting evaluations $A_l$ and $\boldsymbol{b}_l$ are considered erroneous if $A_l\f(\alpha_l) \neq g(\alpha_l)\boldsymbol{b}_l$.
In this scenario they proved that with
\begin{equation}\label{kaltoNumbOfPoints}
L \geq L_{BK} := df+dg+2e+t+1 
\end{equation}
number of points, it is possible to uniquely reconstruct the solution of the linear system (\ref{linSystem}), where
\begin{itemize}
    \item $df \geq \deg{(\f)}:=\max_{1 \leq i \leq n }\deg{(f_i)}$,
    \item $dg \geq \deg{(g)}$,
    \item $e \geq |E|$ is a bound on the erroneous evaluations where
      $$
      E:=\{l \mid A_l\f(\alpha_l) \neq g(\alpha_l)\boldsymbol{b}_l\}
      $$
    \item $t \geq |R|$ is a bound on the rank drops where
    $$
    R:=\{l \mid A_l \f(\alpha_l)=g(\alpha_l)\boldsymbol{b}_l \text{ and }\rank{(A_l)} < n\}
    $$
\end{itemize}

Their method consists in solving the homogeneous linear system
\begin{equation}\label{systemKalt}
\left[A_l \begin{pmatrix}
    \varphi_1(\alpha_l)\\
    \varphi_2(\alpha_l)\\
    \vdots \\
    \varphi_n(\alpha_l)\\
    \end{pmatrix} -\psi(\alpha_l)\boldsymbol{b}_l=0\right]_{l \in \{1, \ldots, L\}}
\end{equation}
where 
\begin{itemize}
    \item $\boldsymbol{\varphi}=(\varphi_1, \ldots, \varphi_n) \in (\K[x])^n$ and for any $1 \leq i \leq n $, $\deg(\varphi_i) \leq df+e$,
    \item $\psi\in \K[x]$, $\deg(\psi)\leq dg+e$.
\end{itemize}
The unknowns of the linear system (\ref{systemKalt}) are the coefficients of $\varphi_i$ and $\psi$. 
In particular, they proved the following:
\begin{theorem}[\cite{Kalt1}]
 Under previous assumptions, let $(\boldsymbol{\varphi}_{min}, \psi_{min})$ be the solution of (\ref{systemKalt}) of minimal degree and $\psi_{min}$ monic. Then
$$
\begin{array}{ll}
     \boldsymbol{\varphi}_{min}=\Lambda \f, &\psi_{min}=\Lambda g\\
\end{array}
$$
where 
$$
\Lambda(x)=\prod_{l \in E}(x -\alpha_l)
$$
is the error locator polynomial.
\end{theorem}

\noindent
From now on we omit the rank drops study and we assume $t=0$.

This method is a generalization of the Berlekamp-Welch decoding algorithm for Reed-Solomon codes \cite{bw}. In fact, if $\K=\Fq$, $m=n=1$, $A_l=I_1$ and $g$ is the constant polynomial $1$, then  
$$
\begin{cases}
b_l=f(\alpha_l) & l \notin E,\\
b_l\neq f(\alpha_l) & l \in E.\\
\end{cases}
$$

\noindent
Hence, in this case, the problem of recovering the solution of the linear system (\ref{linSystem}) with errors, coincides with the problem of decoding of RS code  and 
the linear system (\ref{systemKalt}) is exactly the \textit{key equation} of the classical Berlekamp-Welch method.

We now define the Interleaving RS encoding procedure.

\noindent
Let $\mathcal{C}$ be an $[n,k]_q$ RS code.
\begin{itemize}
    \item we consider $r$ codewords $c_i \in \mathcal{C}$. For any $i \in \{1, \ldots, r\}$, $c_i=(f_i(\alpha_1), \ldots, f_i(\alpha_n))$ where $f_i \in \Fqx$ has degree $\deg(f_i)\leq k-1$ and $\{\alpha_1, \ldots, \alpha_n\}$ is the set of distinct evaluation points;
    \item we arrange these codewords row-wise and we obtain the $r \times n$ matrix $(c_i)_{1\leq i \leq r}=(f_i(\alpha_j))_{\substack{1 \leq i \leq r\\ 1 \leq j \leq n}}$;
    \item by interpreting this matrix as a row vector $(\f(\alpha_j))_{1\leq j \leq n} \in (\mathbb{F}_{q^r})^n$, we obtain a codeword of an Interleaved RS code of length $n$, dimension $k$ over $\mathbb{F}_{q^r}$. The number of codewords $r$ is the amount of interleaving.
\end{itemize}

\begin{definition}[Simultaneous polynomial reconstruction \cite{SPR}]
Let $n,k,e \in \mathbb{N}$ and $\alpha_1, \ldots, \alpha_n$ distinct points in $\Fq$. An instance of the SPR is $(y_{ij})_{\substack{1 \leq i \leq r\\ 1 \leq j \leq n}}$, that verifies the following. There exists 
\begin{itemize}
    \item $E \subset \{1, \ldots, n\}$ with $|E| \leq e$,
    \item polynomials $(f_1, \ldots, f_r)$, with $\deg(f_i) \leq k-1$
\end{itemize}
such that
$$
y_{ij}=f_i(\alpha_j), j \notin E\\
$$
The solution of the SPR is the tuple $(f_1, \ldots, f_r)$.
\end{definition}

\noindent
The SPR problem is exactly the problem of decoding an Interleaved RS code with length $n$, dimension $k$ and amount of interleaving $r$.

We can now observe that,
\begin{remark}\label{remIdentity}
If $\K=\Fq$, $m=n$, $A_l=I_n$ and $g$ is the constant polynomial $1$, then  the linear system (\ref{linSystem}) becomes
$$
\begin{cases}
\boldsymbol{b}_l=\f(\alpha_l) & l \notin E,\\
\boldsymbol{b}_l\neq \f(\alpha_l) &l \in E.\\
\end{cases}
$$
Hence, the problem of solving the linear system (\ref{linSystem}) with errors coincides with the problem of decoding an \textit{Interleaved RS code} (SPR \cite{SPR}) with length $L$, dimension $df+1$ and amount of interleaving $n$.
\end{remark}
In \cite{SPR}, the authors proposed an algorithm that, under some  hypotheses on the error distribution, allows to decode an Interleaved RS code with a certain probability.
In particular, they introduced \textit{key equations},
\begin{equation}\label{sprEqu}
\begin{cases}
[m_1(\alpha_j)=y_{1j}E(\alpha_j)]_{1 \leq j \leq n}\\
\ldots\\
[m_r(\alpha_j)=y_{rj}E(\alpha_j)]_{1 \leq j \leq n}\\
\end{cases}
\end{equation}
The unknowns of this linear system are the coefficients of $m_i$ and $E$, polynomials of degrees at most respectively $k+e$ and $e$.
This linear system (\ref{sprEqu}) has $rn$ equations and $r(k+e)+e+1$ unknowns. 
Moreover, 
\begin{theorem}[\cite{SPR}]
Let $(y_{ij})_{{\substack{1 \leq i \leq r\\ 1 \leq j \leq n}}}$ the received word of an Interleaved RS code, or equivalently an instance of the SPR problem, where 
$$
e:=|E| \leq \frac{r(n-k)}{r+1}
$$
and for each $i \in \{1, \ldots, r\}$,
\begin{enumerate}[(i)]
  \item \label{probAss2} if $j \in E$, $y_{ij}$ are uniformly distributed over $\Fq$
    \item \label{probAss3} if $j \notin E$, $y_{ij}=f_i(\alpha_j)$ and $f_1, \ldots, f_r$ are uniformly distributed over the vector space of polynomials of $\Fqx$ of degree at most $k-1$;
\end{enumerate}
then the linear system (\ref{sprEqu}) admits at most one solution with probability at least $1-e/q$.
\end{theorem}

\noindent
Since if $r \geq 1$, 
$$
\frac{r(n-k)}{r+1} \geq \frac{n-k}{2},
$$
then they proved that, under the probabilistic assumptions (\ref{probAss2}) and (\ref{probAss3}), it is possible to correctly decode the received word beyond the unique decoding bound. The failing probability, \textit{i.e}. the probability that the algorithm fails, is then upper bounded by $e/q$.

In a following paper \cite{Choko}, the probabilistic assumptions are reduced and it was proved that the failing probability is $\mathcal{O}(1/q)$ which is independent of the number of errors.
In detail,
\begin{theorem}[\cite{Choko}]
Given $(y_{ij})_{{\substack{1 \leq i \leq r\\ 1 \leq j \leq n}}}$ the received word of an Interleaved RS code, where
$$
e:=|E| = \frac{r(n-k)}{r+1}
$$
and for each $i \in \{1, \ldots, r\}$,
\begin{enumerate}[(i)]
  \item if $j \in E$, $y_{ij}$ are uniformly distributed over $\Fq$
    \item if $j \notin E$, $y_{ij}=f_i(\alpha_j)$, 
\end{enumerate}
then the linear system (\ref{sprEqu}) admits at most one solution with probability at least $1-\frac{exp(1/(q^{r-2}))}{q-1}$.
\end{theorem}

\noindent

\noindent
In this paper, starting from Remark \ref{remIdentity} we reexamine the problem of solving the linear system (\ref{linSystem}) with errors as a generalization of the decoding of an Interleaved RS code under some hypotheses on the error distribution. In particular, following the \cite{SPR} approach, we prove that the failing probability is at most $\frac{dg+e}{q}$ (where $dg$ is the degree of the common denominator of the vector of rational functions). We stress out that, in our scenario, we relax the hypotheses on the error distribution as in \cite{Choko}.
However we are not able to prove a bound on the failing probability as tight as \cite{Choko} (see Section~\ref{sec:XP} for more comments).

\section{Generalization of decoding of Interleaved RS codes}

We study the problem of solving a consistent, full rank, linear system (\ref{linSystem}), $ A(x)\boldsymbol{y}(x)=\boldsymbol{b}(x)$
with polynomial entries over a finite field $\Fq$. 

\noindent
In a first instance, we focus on the square case, i.e. $n=m$. Let $\y=\frac{\f(x)}{g(x)}$ be the reduced unique solution as in (\ref{gcd}).

\iffalse
We slightly modify the black box scenario described in the previous section, by introducing a probabilistic assumption. 

\noindent
In particular, we assume that, fixed a certain number $L$ of evaluation points, for any erroneous evaluation, i.e. $\alpha_l$ such that $A_l \f(\alpha_l) \neq g(\alpha_l)\boldsymbol{b}_l$, the entries of $A_l$ and $\boldsymbol{b}_l$ are uniformly random elements of $\Fq$.

We define 

\begin{equation} \label{numPoints}
L:=\frac{n(df+e+1)+dg+e}{n}
\end{equation}

where, 
\begin{itemize}
    \item $df \geq \deg{(f)}:=\max_{1 \leq i \leq n }\deg{(f_i)}$,
    \item $dg = \deg{(g)}$,
    \item $e$ is the number of erroneous evaluations
    $$
    \begin{array}{l}
    e = |E|=|\{l \in \{1, \ldots, L\} \mid A_l\f(\alpha_l) \neq g(\alpha_l)\boldsymbol{b}_l\\ \text{ and } \rank(A_l)=n\}|.
    \end{array}
    $$
\end{itemize}
\fi   
We fix $L$ evaluation points with
\begin{equation} \label{numPoints}
L \geq L_{GLZ}:=\left\lceil \frac{n(df+e+1+dg)+e}{n}\right\rceil
\end{equation}
where
\begin{itemize}
    \item $df \geq \deg{(\f)}:=\max_{1 \leq i \leq n }\deg{(f_i)}$,
    \item $dg = \deg{(g)}$,
    \item $g(\alpha_l)\neq 0$ for $1\leq l \leq L$,
    \item $e = |E|$ is the number of erroneous evaluations where
    $$
    E = \{l  \mid A_l\f(\alpha_l) \neq g(\alpha_l)\boldsymbol{b}_l\\ \text{ and } \rank(A_l)=n\}.
    $$
\end{itemize}

\begin{remark}\label{lessNumbOfPoints}
In this way, since $n \geq 1$, we reduce the number of evaluation points, 
$$
L_{GLZ} \leq L_{BK}.
$$
\end{remark}

We slightly modify the black box scenario described in the previous section, by introducing a probabilistic assumption. More specifically, we assume that for any erroneous evaluation, $l \in E$, the entries of $A_l$ and $\boldsymbol{b}_l$ are uniformly random elements of $\Fq$. Moreover, since we skip the rank drops study, we also suppose that all the $A_l$ are always full rank.

We study, for any $l \in \{1, \ldots, L\}$, the homogeneous linear systems

\begin{equation}
    A_l\boldsymbol{\gamma}_l-\sigma_l \boldsymbol{b}_l=0.
\end{equation}

\begin{remark}\label{squareCase}
Let $C_l$ be the coefficients matrix of any of these linear systems,
$$
C_l=[A_l|-\boldsymbol{b}_l].
$$ Since any system have $n$ equations and $n+1$ unknowns and $A_l$ is full rank, the rank of $C_l$ is $n$ and in particular, the kernel of any $C_l$ is one dimensional.
\end{remark}

\begin{proposition}
Let $(\boldsymbol{\gamma}_l, \sigma_l)=(\gamma_{l1}, \ldots, \gamma_{ln}, \sigma_l)$ be the vector that generates the right kernel of $C_l$, with $l \in \{1, \ldots, L\}$. Then,
$$
\begin{array}{ll}
\frac{\boldsymbol{\gamma}_l}{\sigma_l}=\frac{\f(\alpha_l)}{g(\alpha_l)}& \forall l \notin E.
\end{array}
$$
\end{proposition}
\begin{proof}
By our assumptions, $g(\alpha_l) \neq 0$ for any $l \notin E$.
Now, since the right kernel is one dimensional, its generator $(\boldsymbol{\gamma}_l, \sigma_l)=(\gamma_{l1}, \ldots, \gamma_{ln}, \sigma_l)$ is a nonzero vector, and also $\sigma_l \neq 0$.
Let $l \notin E$, be a correct evaluation. Since $A_l\boldsymbol{\gamma}_l-\sigma_l \boldsymbol{b}_l=0$ and $A_l\f(\alpha_l)=g(\alpha_l)\boldsymbol{b}_l$, then $A_l(\f(\alpha_l)\sigma_l-g(\alpha_l)\boldsymbol{\gamma}_l)=0$. The matrix $A_l$ is full rank, hence $\f(\alpha_l)\sigma_l-g(\alpha_l)\boldsymbol{\gamma}_l=0$. Then we can conclude that
$\frac{\boldsymbol{\gamma_l}}{\sigma_l}=\frac{\f(\alpha_l)}{g(\alpha_l)}$.
\end{proof}

\noindent
Following the previous notations, if we denote by $\boldsymbol{y}_l:=\frac{\boldsymbol{\gamma}_l}{\sigma_l}=\frac{1}{\sigma_l}(\gamma_{l1}, \ldots, \gamma_{ln})\in (\Fq)^n$, for any $l \in \{1, \ldots, L\}$ we have that
\begin{equation}\label{errEva}
\boldsymbol{y}_l=\frac{\f(\alpha_l)}{g(\alpha_l)} ,l \notin E
\end{equation}
\begin{remark}\label{probAss}
By our probabilistic assumption, $(y_{li})_{l\in E}$ are uniformly random elements of $\Fq$.
\end{remark}

\noindent
In this way, we reduce the problem of solving the linear system (\ref{linSystem}) to the reconstruction of a vector of rational functions with errors.
In particular, we observe that if $g$ is the constant polynomial $1$, our problem coincides exactly with the decoding of an Interleaved RS code, with length $L$, dimension $df+1$ and amount of interleaving $n$.
This is why we can consider our problem as a generalization of the decoding of Interleaved RS codes.

Now, we study the key equations
\begin{equation*}
\begin{cases}
\boldsymbol{\varphi}(\alpha_1)=\boldsymbol{y}_1\psi(\alpha_1)\\
\ldots\\
\boldsymbol{\varphi}(\alpha_L)=\boldsymbol{y}_L\psi(\alpha_L)\\
\end{cases}
\end{equation*}

\noindent
or, in other terms, if we denote $\boldsymbol{y}_l=(y_{l1},\ldots, y_{ln})$ for $l \in \{1, \ldots, L\}$,
\begin{equation}\label{ourKeyEqu}
\begin{cases}
[\varphi_1(\alpha_l)=y_{l1}\psi(\alpha_l)]_{1\leq l \leq L}\\
\ldots\\
[\varphi_n(\alpha_l)=y_{ln}\psi(\alpha_l)]_{1\leq l \leq L}\\
\end{cases}
\end{equation}

\noindent
where 
\begin{itemize}
    \item $\boldsymbol{\varphi}=(\varphi_1, \ldots, \varphi_n) \in (\Fqx)^n$ and $\deg(\varphi_i) \leq df+e$,
    \item $\psi \in \Fqx$ has degree at most $dg+e$ .
\end{itemize}

\noindent
The linear system (\ref{ourKeyEqu}) has $nL$ equations and $n(df+e+1)+dg+e+1$ unknowns, which are the coefficients of $\boldsymbol{\varphi}$ and $\psi$.

\noindent
The coefficient matrix of the system $(\ref{ourKeyEqu})$ is
$$
M_{\y} := \left(\begin{array}{lllc}
V_{df+e+1} &           &          &-D_{1}V_{dg+e+1} \\
           & \ddots    &          &\vdots \\
           &           &V_{df+e+1}&-D_{n}V_{dg+e+1} 
\end{array}\right)
$$
where, 
\begin{itemize}
    \item  $V_t=(\alpha_l^{i-1})_{\substack{1 \leq l \leq L\\ 1\leq i \leq t}}$ is the $L \times t$ \textit{Vandermonde} matrix, 
    \item for $i \in \{1, \ldots, n\}$,
    $D_{i}$ is the diagonal matrix with $y_{1i}, \ldots, y_{Li}$ on the diagonal.
\end{itemize}
Recall that the error locator polynomial
$
\Lambda(x)=\prod_{l \in E}(x-\alpha_l)
$ 
is monic and has degree $e$.  We observe that
$(\Lambda\f, \Lambda g)=(\Lambda f_1, \ldots, \Lambda f_n, \Lambda g)$ is a
solution of the system. Therefore, if the kernel of $M_{\y}$ has dimension 1, the
non-zero solutions are collinear to $(\Lambda\f, \Lambda g)$, meaning that we can
correctly reconstruct the fraction $\f /g = \Lambda\f / \Lambda g$.

We do not have \emph{a priori} any other information about this kernel. In our
following main result, we adapt the approaches of \cite{SPR, Choko} to prove that this favorable situation happens with high probability.

\begin{theorem}\label{mainRes}
  Under the previous assumptions, the dimension of the (right) kernel of $M_{\y}$ is
  one with probability at least $1-\frac{(dg+e)}{q}$.
\end{theorem}

The cornerstone of the proof is the following lemma.

\begin{lemma}\label{mainLemma}
There exists a random draw of $(\y_l)_{l \in E}$ such that the dimension of the right kernel of $M_{\y}$ is one.
\end{lemma}

\begin{proof}
  We can partition $E = \bigcup_{1 \leq i \leq n}I_i$ with sets $I_i \subset E$
  such that $|I_i| \leq L - (df+dg+e+1)$ for
  $1 \leq i \leq n$ since $n(L-(df+dg+e+1)) \geq e$. We start by
  fixing a part of the random variables $(\y_l)_{l \in E}$ : for all $1 \leq i \leq n$ and
  $l \in E \setminus I_i$, we set $y_{li}=\frac{f_i(\alpha_l)}{g(\alpha_l)}$
  while $y_{li}$ for $l \in I_i$ remain free variables for now, for a total of
  $e$ free variables.

  Now, we study the equations (\ref{ourKeyEqu}). Fixed $1 \leq i \leq n$,
 for $l \notin I_i$, $\varphi_i(\alpha_l)g(\alpha_i)=f_i(\alpha_l)\psi(\alpha_l)$.
    Therefore, since the polynomial $\varphi_ig-f_i\psi$ has degree at most $df+dg+e$ and at least $df+dg+e+1$ roots, it is the zero polynomial.
    Hence, the two fractions are equals,
    $$
    \frac{\boldsymbol{f}}{g}=\frac{\boldsymbol{\varphi}}{\psi}
    $$
    Moreover, since the fraction $\frac{\boldsymbol{f}}{g}$ is reduced, there exists a polynomial $R$ such that $f_1R=\varphi_1$, $gR=\psi$ and  $\deg(R) \leq e$.

  Hence for $1\leq i\leq n$, $f_iR=\varphi_i$ and $gR=\psi$ and so by replacing in
  (\ref{ourKeyEqu}) we have
  $$
  R(\alpha_l)[f_i(\alpha_l)-y_{li}g(\alpha_l)]=0
  $$
  We observe that for any $l \in E$, there exists $1\leq i\leq n$ such that
  $l \in I_i$, so $y_{li}$ is still a free variable and we can give it a value
  so that $f_i(\alpha_l) \neq y_{li}g(\alpha_l)$ since $g(\alpha_l) \neq
  0$. Therefore $R(\alpha_l) = 0$ for $l \in E$ and $deg(R) \leq e$ and so $R$
  is a scalar multiple of $\Lambda$. We have proved that for some values of
  $\y_l$ the kernel is spanned by $(\Lambda \f, \Lambda g)$ and has dimension 1.

\end{proof}

\begin{proof}[Proof of Theorem \ref{mainRes}]
  We recall that $(\Lambda f_1, \ldots, \Lambda f_n, \Lambda g)$ is a solution of the linear system (\ref{ourKeyEqu}) so it has kernel dimension at least 1. Since $\Lambda g$ is a monic polynomial of degree $dg+e$, the last column of $M_{\y}$ is linearly dependent on the previous ones. 
  As a consequence, the kernel of $M_{\y}$ has dimension $1$ iff the rank of $M_{\y}$ is
  $\rho := n(df+e+1) + dg+e$ iff there exists a non-zero minor of $M_{\y}$ of size
  $\rho$ that avoids the last column. Considering the minors as polynomials in the variables $(\y_l)_{l \in E}$, we have shown in Lemma~\ref{mainLemma} that one of these $\rho$-minors is not the zero polynomial because it does not vanish on some value of $(\y_l)_{l \in E}$. 
  Finally, since this $\rho$-minor has degree at most $dg+e$, by Schwartz-Zippel Lemma, it cannot be zero in more
  than a $\frac{dg+e}{q}$-fraction of its domain. Therefore, we can conclude that the kernel
  has dimension 1 with probability at least $1-\frac{dg+e}{q}$.
\end{proof}

Summing up, working under our probabilistic assumptions, we are able to recover
the correct solution with a failing probability upper bounded by
$\frac{dg+e}{q}$.

\begin{algorithm}
 \KwData{$(A_{l},\boldsymbol{b}_{l})_{1 \leq l \leq L}$ and $df, dg, e$}
 \KwResult{$(\f,g)$ or \texttt{fail}}
 
 $L:=\lceil \frac{n(df+e+1+dg)+e}{n}\rceil$\;
 \For{$l=1, \ldots, L$}{
 find a basis $\{(\boldsymbol{\gamma}_l, \sigma_l)\}$ of the right kernel of $C_l$\;
 $\boldsymbol{y}_l:=\frac{\boldsymbol{\gamma}_l}{\sigma_l}$\;
 }
 construct the matrix $M_{\y}$ of the key equation (\ref{ourKeyEqu})\;
 \eIf {$\rank({M_{\y}})=n(df+e+1)+dg+e$}{
        compute a solution $(\boldsymbol{\varphi}, \psi)$ with $\psi$ monic\;
        $\Lambda:=\GCD(\boldsymbol{\varphi}, \psi)$\;
        \textbf{return} $(\frac{\boldsymbol{\varphi}}{\Lambda}, \frac{\psi}{\Lambda})$\;
        }
        { \textbf{return} \texttt{fail}\;}
\end{algorithm}

Up to this point, we have assumed that the linear system (\ref{linSystem}) is square, \emph{i.e}. $n=m$. 
 With our method it is possible to recover the solution with the same probability also in the general case by considering random $y_{li}$, for any $l$ such that $\rank(C_l) = n + 1$ (see Remark \ref{squareCase}).

\section{Experiments and Conclusions}\label{sec:XP}
In this work we prove that, in our probabilistic scenario, by using the evaluation interpolation technique \cite{McClellan} with $L \geq L_{GLZ}$ evaluation points, we can reconstruct the vector solution of the linear system (\ref{linSystem}). 
Recall that $L_{BK}$ is the number of points that guarantees to uniquely reconstruct the solution for every error. In our case since $L_{GLZ} < L_{BK}$, we cannot reconstruct the solution for every error, but for almost all of them.

We implement our algorithm in  \texttt{SageMath} (http://www.sagemath.org).
In particular, we solve $20$ different linear systems with polynomial coefficients, of size $3$. For every system we apply our algorithm $1000$ times, with $e=5$ errors and $dg=2$. We compute the percentage $p$ of the number of times the algorithm fails.
We compare our experimental results with $p_{GLZ}$, \textit{i.e.} the failing probability of Theorem \ref{mainRes}, and with $p_{BMS}$, \textit{i.e.} the failing probability  of \cite{Choko}. We recall that the last one is linked to a related problem, the decoding of Interleaved RS, but not to our problem.  
%In detail, recall that $p_{BMS} = \mathcal{O}(\frac{1}{q})$, which it is independent on the number of errors (and the degree of $g$).
We obtain the following results:
\begin{center}
\begin{tabular}{|c|c|c|c|c|}
     \hline
     $q$ & $ p $  & $p_{GLZ}$ & $p_{BMS}$ & $p^\ast \textrm{ (see below)}$ \\  
     \hline
     $2^4$ & $0.3\%$  & $43.7\%$ & $7.1\%$ &$0.4\%$ \\
     \hline
    $2^5$ & $0.1\%$ & $21.9\%$ & $3.3\%$ & $0.2\%$  \\
     \hline
    $2^6$  & $0$  & $10.9\%$ & $1.6\%$ & $0.1\%$\\
    \hline
\end{tabular}
\end{center}

First of all, these experiments suggest that the failing probability is very small and that it decreases when the cardinality of the field grows.
Moreover, since experiments are significantly under our theoretical bound, our bound could be strongly improved, which we leave to future work. % Future work could consider new techniques to restate that bound. 
This may also suggest that our bound $L_{GLZ}$ is not optimal.
Therefore, we have done some experiments using an inferior number of evaluation points
$$
L = \left\lceil\frac{n(df+e+1)+dg+e}{n}\right \rceil.
$$
and we note $p^\ast$ the corresponding probability. This new $L$ is chosen in such a way that the coefficient matrix $M_y$ has now $nL$ equations and $nL-1$ unknowns so that its kernel might have dimension 1. 
We can observe that also in this case, we have a very high probability to reconstruct the correct vector rational solution in practice. Therefore future work could try to extend our results to this new number of evaluation points.

Finally we can remark that, since we use an evaluation interpolation technique, there exist finitely many evaluation points $\alpha_l$ such that the evaluated matrix $A(\alpha_l)$ is not full rank anymore. In this work, we have chosen the evaluation points such that the evaluated matrix is full rank, thus omitting the rank drop analysis, which we leave again to future work.
% Therefore, this work can be extended for dealing with rank drops by using the main idea in \cite{Kalt1}.

\section{Acknowledgements}
This work has been supported by Occitanie Region through the ARPE HPAC project. We would also like to thank Daniel Augot and Bruno Grenet for useful discussions.

\end{document}